\documentclass[12pt]{article}
\usepackage{graphicx,amssymb,amsmath,amsthm,cite}

\newtheorem{proposition}{Proposition}
\newtheorem{property}{Property}
\newtheorem{remark}{Remark}

\def\be{\begin{equation}}
\def\ee{\end{equation}}
\def\ben{\begin{eqnarray}}
\def\een{\end{eqnarray}}

\newcommand{\la}{\langle}
\newcommand{\ra}{\rangle}

\newcommand{\til}{\tilde}

\newcommand{\kfV}{|f_{\SV} \ra}

\newcommand{\kfWC}{|f_{\SWC} \ra}

\newcommand{\kvi}{|v_i\ra}
\newcommand{\bvi}{\la v_i |}
\newcommand{\kwi}{|w_i\ra}
\newcommand{\bwi}{\la w_i |}
\newcommand{\bui}{\la u_i |}
\newcommand{\kui}{|u_i\ra}

\newcommand{\wt}{{w}}

\newcommand{\wtkp}{{w}^{k+1}}
\newcommand{\wtik}{\wt_i^k}
\newcommand{\cik}{c_i^k}
\newcommand{\wtikp}{\wt_i^{k+1}}

\def\kj{k \setminus j}
\def\rj{r \setminus j}

\def\op{\hat{P}}

\def\S{{\cal{S}}}
\def\SS{{\cal{S}}}

\def\SH{{\cal{H}}}
\def\SV{{\cal{V}}}
\def\SW{{\cal{W}}}

\def\SVKF{{\cal{V}}_K}

\def\SWKP{{\cal{W}}_{k+1}}
\def\SVKJ{{\cal{V}}_{\kj}}
\def\SWKJ{{\cal{W}}_{\kj}}
\def\SWKrJ{{\cal{W}}_{\rj}}
\def\SW{{\cal{W}}}
\def\SWK{{\cal{W}}_k}
\def\SWKr{{\cal{W}}_r}
\def\SWC{{\cal{W}^\bot}}

\def\EVW{\hat{E}_{\SV \SWC}}

\def\EVKW{\hat{E}_{\SV_k \SWC}}

\newcommand{\Spann}{\text{span}}

\def\C{\mathbb{C}}
\def\emptyy{\{0\}}

\newcommand{\ut}{w}

\newcommand{\utik}{w_i^k}

\title{Measurements design and phenomena discrimination}

\author{Laura Rebollo-Neira\\
Mathematics\\
Aston University\\ 
Birmingham, B4 7ET, UK}

\begin{document}
\maketitle

\begin{abstract}
The construction of measurements suitable for discriminating 
signal components produced by phenomena of different nature is 
considered. The required measurements should be capable 
of cancelling out those signal components which are to be ignored 
when focussing on a phenomenon of interest. 
Under the hypothesis that the subspaces hosting the 
signal components produced by each phenomenon are complementary,  
their discrimination is accomplished by measurements giving rise to 
the appropriate oblique projector operator. 
The subspace onto which the operator should 
project is selected by non-linear 
techniques in line with adaptive pursuit strategies.
\vspace{1cm}


\end{abstract}

\section{Introduction}
The word {\em{signal}} is frequently used to refer to a physical carrier  
convening information about some phenomenon. We adopt 
such terminology and further refer to the process of 
transforming a signal into a number (within the corresponding units) 
as a  {\em{measurement}}.
An appropriate mathematical setting for this description is to 
consider that a signal is an element of some vector space 
and a measurement a
functional transforming the vector into a scalar. In this effort we  
discuss the design of measurements in relation to the following problem:
Assume that a signal $f$, represented 
as an element of an inner product space $\SH$, arises 
by the superposition of two components, $f_1$ and $f_2$, 
each of which is produced by a particular phenomenon   
and such that $f_1 \in \S_1$ and  $f_2 \in \S_2$, 
where $\S_1$ and $\S_2$ are disjoint subspaces of $\SH$, i.e. 
$\S_1 \cap \S_2 = \emptyy$. This condition implies that the 
superposition $f=f_1 + f_2$ is unique. The matter to be 
addressed here concerns the construction of the appropriate measurements  
allowing us to discriminate the component, say $f_1$, from the 
available signal $f$ and the knowledge of $\S_1$ and $\S_2$. Under 
the condition $\S_1 \cap \S_2 = \emptyy$  the problem has 
a straightforward `theoretical' 
solution, since the component $f_1$ can be  extracted from $f$  by an 
oblique projection onto $\S_1$ and along $\S_2$ \cite{BS94,Reb07a}. 
Unfortunately,  
even when theoretically the condition $\S_1 \cap \S_2 = \emptyy$ is 
satisfied, if the subspaces $\S_1$ and $\S_2$ are not well separated, 
the construction of the corresponding oblique projector 
becomes ill posed. Consequently, the signal splitting can not be achieved 
by numerical calculations in finite precision arithmetics. 
This is the situation we are concerned with. {\em{We assume that the given
subspaces $\S_1$ and $\S_2$ are `theoretically' disjoint,  but
close enough to yield an ill posed problem}}.

Our proposal for the numerical realization of the phenomenon discrimination 
is focussed on the search of a subspace of the given
$\S_1$, where a class of signals is considered to lie. It will be 
assumed throughout the paper that the class of signals to be 
considered is $K$-sparse in a spanning set for
$\S_1$. By this we mean that given a spanning set for $\S_1$, 
the corresponding linear superposition of a signal  
 has at most $K$ nonzero coefficients. The $K$-value should 
be less than or equal to the dimension of the subspace  of  
$\S_1$ such that the construction of measurement vectors 
giving  rise to an oblique projection onto 
itself, and along $\S_2$, is well conditioned. This assumption 
is quite realistic, considering that in practice there is often  a
lack of complete knowledge on the actual subspace 
$\S_1$ and to be on the safe side one may overestimate it.  However, 
the assumption does not make the problem much easier to solve. Indeed, 
the problem of subspace selection is in  general a combinatorial 
problem,  whereby an exhaustive search of possibilities 
is in general intractable. 
The approach we propose in this Communication evolves by step wise optimal 
selection and is in line with the adaptive greedy approximation termed
Matching Pursuit.
 Such a technique, which appeared first in 
 the statistic literature \cite{Jon87},  
has been extended 
in the area of signal processing to several greedy 
strategies \cite{MZ93,DMA97,RL02,ARS04,Tro04,AR06} 
being currently of 
 assistance to a range of disciplines, including physics \cite{WB03,WB04,Jie08}.
 In particular, we revise and extend the Oblique Matching 
Pursuit (OBMP) approach which has been recently proposed in 
relation the above described problem of signal discrimination 
\cite{Reb07b}. 

The paper is organized as follows: In Section \ref{sec2} we 
introduce the mathematical setting for signal representation to 
be adopted here and in Section \ref{sec3} with discuss
the construction of oblique projectors. Section 
\ref{sec4} highlights the importance of
the search for sparse representations in the construction of oblique 
projectors for phenomenon discrimination. The proposed strategy 
is discussed in Section \ref{sec5} and illustrated in Section 
\ref{sec6} by two numerical simulations: 1)The 
cancellation of impulsive noise from the register of a system of 
harmonic oscillators and 2) the
separation of a spectrum from blackbody radiation background.
The conclusions are presented in Section \ref{sec7}. 
                                        
\section{Mathematical framework}
\label{sec2}
Regardless of the informational content of a signal,  we deal 
with it mathematically by considering it as an 
element of an inner product space $\SH$. Thus, adopting 
Dirac's  notation, we represent a signal $f$ as a ket $|f \ra$ and 
the corresponding dual as a bar $ \la f |$. Accordingly, 
the square norm $|||f \ra ||^2$ is induced by the 
inner product that we indicate as $\la f |f \ra$.
For our present purpose we further assume 
that all the signals of interest belong to some finite 
dimensional subspace $\SV \in \SH$ spanned by a 
finite set $\{\kvi \in \SH\}_{i=1}^M$. Consequently,   
for every signal $\kfV \in \SV$ there exists a set of numbers 
$\{c_i\}_{i=1}^M$ which allow us to 
express the signal as the linear superposition 
\be
\label{mod}
\kfV = \sum_{i=1}^M c_i |v_i\ra.
\ee
In the jargon of signal processing the above expansion is 
called {\em atomic decomposition} and the vectors in the 
decomposition are called `atoms' \cite{Mal99}. 
In applications where an economical signal representation is 
important, the goal is to construct decompositions 
involving as few terms as possible. For this  end the atoms 
are selected from a large and,  in general redundant, set 
called a {\em dictionary}. If the number of $M$-terms 
in \eqref{mod} is satisfactorily small in relation with the particular 
application, the decomposition is said to be {\em{sparse.}}

Even when we think of a signal as an abstract object   
in an inner product space, for processing tasks we need a numerical    
representation of such an object. The process of 
transforming the signal into a number is refereed 
to as  a {\em {measurement}} or {\em {sampling}}.
Since we restrict considerations to linear measurements,
we represent them by {\em {linear functionals}}. 
Thus, making use of Riesz' theorem \cite{RS80} we can  express a linear 
measurement as
$m= \la  w | f \ra $ for some $|w \ra\in \SH.$  Hence, 
considering $M$ measurements $m_i,\, i=1,\ldots,M$,
each of which is obtained by a {\em{measurement vector}} $\kwi$
we have a numerical representation of the ket $ |f \ra$ as given by
\be
\label{mset}
m_i= \la  w_i | f \ra,\quad i=1,\ldots,M.
\ee
The representation of measures  as in   
(\ref{mset}), which have been used in physics for many years, 
has started to become popular within other disciplines
through the theory of Compressed Sensing 
\cite{Don06,CR06,Ct06,Bar07,CSWebpage}.  

The problem of reconstructing the whole information content 
of a signal from a numerical 
representation has been extensively studied for the last 
thirty years from a number of different 
points of view. The diversity of available approaches 
is specially helpful
when the reconstruction is to be achieved on the basis of
incomplete information. In particular, the  above mentioned 
theory of Compressed Sensing has produced strong theoretical 
results with regard to the recovery of a signal, assumed to be 
sparse in some orthonormal basis, from a number of 
non adaptive measurements which can be significantly less than the dimension of 
the signal subspace \cite{Don06,CR06,Ct06,Bar07,CSWebpage}. Here 
we focus on the particular problem involving 
the reconstruction of a signal 
from a set of linear measurements as given in (\ref{mset}), 
with no further numerical calculations other 
than using these numbers in the expansion (\ref{mod}). 
In other words, {\em we wish to use the measurements as coefficients in the 
linear combination (\ref{mod}}). The question then arises as 
to which are the conditions to be requested for the measurement 
vectors $\kwi$ to produce the corresponding 
numbers allowing  the reconstruction of the signal $ |f \ra$ as 
\be
\label{mod2}
\kfV = \sum_{i=1}^M |v_i\ra \bwi f\ra. 
\ee
A major consequence of working under the assumption that 
the signal of interest, $\kfV$, belongs  
to a {\em finite} dimensional subspace, $\SV$, is {\em the lack of 
uniqueness of the measurement vectors $\{\kwi\}_{i=1}^M$, 
even when the spanning set $\{\kvi\}_{i=1}^M$ is linearly independent}.
This statement appears clearly from the following observation. 

Let us denote $\sum_{i=1}^M |v_i\ra \bwi$ as an operator, 
$\hat{E}$, so as to recast \eqref{mod2} in the fashion
\be
\label{mod2c}
\kfV = \hat{E} |f\ra.
\ee
This equation tells us that the measurement vectors 
$\{\kwi\}_{i=1}^M$ should be such that
operator $\hat{E}$ is a projector onto $\SV$. Indeed, 
the operator $\hat{E}$ is a projector if and only if it is
{\em{idempotent}} \cite{RS80,Gal04}, i.e.,
$\hat{E}^2= \hat{E}.$
Consequently, as discussed below, the projection is onto
the range of the
operator, ${\cal{R}}(\hat{E})$,  and along its null space
${\cal{N}}(\hat{E})$. 

Denoting by ${\cal{D}}$ the domain of $\hat{E}$ we recall that
$${\cal{R}}(\hat{E}) =\{ |f \ra,\,\text{such that}\,
 |f \ra= \hat{E}|g\ra,\, \text{for some} |g\ra \in {\cal{D}}\}.$$
Thus, for $\hat{E}$ an idempotent operator 
and for $| f \ra  \in {\cal{R}}(\hat{E})$, 
we have
$\hat{E} | f \ra =\hat{E}^2 |g\ra =\hat{E}|g\ra= | f \ra.$
This implies
that $\hat{E}$ behaves like the identity operator for all
$| f \ra \in {\cal{R}}(\hat{E})$, regardless of ${\cal{N}}(E)$, which is
defined as
$${\cal{N}}(E)= \{|g\ra,\,\text{such that}\,  \hat{E}|g\ra=0,\, |g\ra\in
{\cal{D}}\}.$$
It is now clear that to reconstruct a signal $|f \ra \in \SV$ by means of
(\ref{mod2}) the measurement vectors $\{\kwi\}_{i=1}^M$
should give rise to an operator $\sum_{i=1}^M |v_i\ra \bwi$, which must be 
a projector onto $\SV$.
It is appropriate to point out that the required operator 
is {\em not unique} (even if 
the spanning set $\{\kvi\}_{i=1}^M$ is linearly independent) 
because there exist
many projectors onto $\SV$ having different ${\cal{N}}(\hat{E})$.
Consequently, for reconstructing
signals in the range of the projector its null space
can be chosen arbitrarily. Nevertheless, the
null space, and therefore the particular 
measurement vectors, become crucial when the
projector is to be applied on signals outside its range. It 
follows then that the measurement vectors 
  $\{\kwi\}_{i=1}^M$  can be tailored for a particular 
purpose. Such a degree of freedom will be 
indicated hereforth by using two subscripts for representing a projector. 
We adopt the notation $\hat{E}_{\SV \SWC}$ to indicate a 
projector onto the subspace $\SV$ and along the subspace $\SWC$. 
The particular case
$\hat{E}_{\SV \SV^\bot}$, where $\SV^\bot$ is orthogonal
to $\SV$, corresponds to an {\em {orthogonal projector}}
and we use the special notation $\hat{P}_{\SV}$ to denote  such a
projector. The orthogonal projector is 
popular in approximation techniques because if a signal 
$| f \ra $ is to be approximated by a signal $|f_{\SV}\ra \in \SV$ the
choice $\kfV = \hat{P}_{\SV} | f \ra$ is known to yield the 
unique signal in $\SV$ minimizing the distance 
$||| f \ra - |f_{\SV}\ra||$. 
However, as discussed below, if one is interested in discriminating 
signal components produced by phenomena of different nature an alternative 
 selection of the subspace $\SWC$ is required. 
When
$\SWC$ is not orthogonal to ${\SV}$ the projector $\hat{E}_{\SV \SWC}$
is referred to as an {\em oblique projector}.  

Let as assume for instance that a signal $ |f \ra$ is the superposition of
two signals, $|f \ra=\kfV+ \kfWC$, each component being produced by a different 
 phenomenon we wish to discriminate. Let us assume further that we can 
model the subspaces $\SV$ and $\SWC$ hosting each signal component and 
such subspaces are disjoint, i.e. 
$\SV \cap \SWC = \emptyy$. Thus we can obtain  $\kfV$   from
$| f \ra$, by an oblique projector onto $\SV$ and along $\SWC$. The
projector will map to zero the component $\kfWC$ to produce
$$\kfV= \hat{E}_{\SV \SWC} |f \ra.$$
In the next section we discuss the construction of measurement vectors 
$\{\kwi\}_{i=1}^M$ giving rise to the desired projector. 
\section{Constructing measurement vectors for discrimination of 
signal components}
\label{sec3}
Given two disjoint subspaces $\SV $ and $\SWC$, 
in order to provide a prescription for
constructing the projector $\hat{E}_{\SV \SWC}$ one can proceed as follows. 
Firstly we define $\SS$ as the direct sum of $\SV$ and $\SWC$,
which we express as
 $$\SS= \SV \oplus \SWC.$$
Let $\SW= (\SWC)^\bot$ be the orthogonal complement of $\SWC$ in $\SS$. Thus
we have
$\SS= \SV \oplus \SWC=\SW \oplus^\bot \SWC,$
where the operation $\oplus^\bot$ indicates the orthogonal sum, which 
refers to the direct sum of orthogonal subspaces. 

Considering that $\{\kvi\}_{i=1}^M$ is a spanning set for
$\SV$, a spanning set for $\SW$ is obtained as
\be
\label{ku}
\kui=\kvi - \op_{\SWC} \kvi = \op_{\SW}\kvi,\,i=1,\ldots,M.
\ee
Denoting as $\{| i \ra \}_{i=1}^M$ the standard orthonormal basis for
$\C^M$, we
define the operators $\hat{V}: \C^M \to \SV$ and $\hat{U}:\C^M \to \SW$
as
$$\hat{V}=\sum_{i=1}^M \kvi  \la i|,
\;\;\;\;\;\;\;\;\;\;\;\;\;\;
\hat{U}=\sum_{i=1}^M  \kui   \la i|.$$
Thus the adjoint operators $\hat{U}^\ast$ and
$\hat{V} ^\ast$ are expressed as
$$\hat{V}^\ast=\sum_{i=1}^M  |i \ra \bvi,
\;\;\;\;\;\;\;\;\;\;\;\;\;\
\hat{U}^\ast=\sum_{i=1}^M  |i \ra  \bui.$$
Notice that
$\op_{\SW} \hat{V}= \hat{U}$ and
$\hat{U}^\ast \op_{\SW} = \hat{U}^\ast$
hence, $\hat{G}: \C^M \to \C^M $ defined as:
$$\hat{G}=\hat{U}^\ast  \hat{V}= \hat{U}^\ast  \hat{U}$$ 
is a self-adjoint operator. Its matrix representation being 
given by the elements $\la i| \hat{G} | j\ra= \la u_i |v_j\ra= \la  u_i |u_j\ra, \, i,j=1,\dots,M.$
\begin{remark}It is appropriate to stress that
\begin{itemize}
\item
Operators $\hat{V}$ and $\hat{U}$ are given in terms of 
spanning sets for the spaces 
$\SV$ and $\SW$, respectively, and  {\em any} such spanning set 
can be used. 
\item
The condition $\SV \cap \SWC = \emptyy$ implies that the dimension of $\SV$
is equal to the dimension of $\SW$. Hence, provided that the 
spanning set $\{\kvi\}_{i=1}^M$ is linearly independent, operator $\hat{G}$
has an inverse. Nevertheless, the independence of $\{\kvi\}_{i=1}^M$ 
is not a requirement 
and therefore an inverse for $\hat{G}$ need not exist. For  
the sake of generality we shall use $\hat{G}^\dagger$,
which indicates a pseudo-inverse of $\hat{G}$.
\end{itemize}
\end{remark}
The oblique projector operator onto $\SV$ and along $\SWC$ is 
given as \cite{Eld03}
\be
\label{yoni}
\EVW= \hat{V} \hat{G}^\dagger \hat{U}^\ast,
\ee 
or, equivalently, as
\be 
\label{obli2}
\EVW= \sum_{i=1}^M \kvi \bwi,
\ee
with
\be
\kwi=  \hat U \hat{G}^\dagger |i \ra= \sum_{j=1}^M = 
  | u_j \ra \la j| \hat{G}^\dagger |i \ra. 
\ee
It is actually  straightforward to verify that $\EVW$ given in 
(\ref{obli2}) satisfies the required properties. Namely, 
$\EVW^2 = \EVW,$  
$\EVW \kfV = \kfV$ for all $\kfV \in \SV$, and 
$\EVW | g  \ra = 0$ or all $ | g \ra  \in \SWC.$
\begin{remark}
%\begin{itemize}
%\item
The construction of 
an oblique projector is similar to that of
an orthogonal one. The difference being that in
general the subspaces
$\Spann \{\kvi\}_{i=1}^M =\SV \quad \text{and}\quad \Spann\{\kwi\}_{i=1}^M= \SW$
are different.
For the special case $\{\kvi\}_{i=1}^M=\{\kui\}_{i=1}^M$ we 
have $\Spann\{\kwi\}_{i=1}^M= \Spann \{\kvi\}_{i=1}^M=\SV$, 
thereby the projector is self adjoint and, 
consequently, an orthogonal projector onto $\SV$ along $\SV^\bot$.
%\item
%In order to render the numerical calculation of the
%vectors $\kwi \,i=1,\ldots,M$ spanning $\SW$ as stable as possible,
%it is convenient to orthonormalize vectors $\kui,\,i=1,\ldots,M$ 
%(c.f. \eqref{ku}) to obtain the vectors $\kqi,\,i=1,\ldots,M$ satisfying 
%$\Span\{\kqi\}_{i=1}^M=\Span\{\kwi\}_{i=1}^M=\SW$ and
%$\la q_i | q_j \ra=\delta_{i,j},\,i,j=1,\ldots,M$. 
%With these vectors we construct the $M\times M$ matrix
%$\ki \hat{G}|j\la, = \la q_i | v_j \ra\, i,j=1,\ldots,M$
%\end{itemize}
\end{remark}
We are already in a position to 
extract the component $\kfV$ from $| f \ra$ by the 
simple operation
$\kfV = \hat{V} \hat{G}^\dagger \hat{U}^\ast | f \ra.$
However, the correct discrimination of the signal components  
is successful provided that the subspaces $\SV$ and 
$\SWC$ are well separated. Unfortunately, this is not
always the case and the construction of the necessary projector
may generate an ill posed problem.  In spite of the
fact that `theoretically' $\SV \cap \SWC = \emptyy$, numerical errors,
due to the existence of small eigenvalues  values of the operator $\hat{G}$, 
may cause the
failure to find the unique signal splitting that theoretically
one should expect.
Nevertheless, the correct separation is still possible, 
provided that the signal $\kfV$  
admits a sparse representation in some spanning set for $\SV$. In order words, 
one could succeed in  extracting $\kfV$, provided that it is 
well represented in a subspace $\SV_K \subset \SV$ inducing a 
subspace $\SW_K  \subset \SW$ (satisfying $\SV_K  + \SWC= \SW_K \oplus^\perp 
\SWC$) where the computation of the measurement vectors is well posed. 
If this is the case, the problem of designing measurement vectors for 
discriminating signal components can be addressed as the problem of 
finding the subspace $\SV_K$ where $\kfV$ is well represented. 
Unfortunately
the search for the subspace $\SV_K$ is in general intractable. Indeed, let us 
assume that $\{|v_i\ra\}_{i=1}^M$ is an spanning set for $\SV$ and $\SV_K$ is 
spanned by $K$ elements of such a set. Even possessing this knowledge,  
the problem of finding the right subspace by exhaustive search 
would be a combinatorial problem: 
Out of a set of cardinality $M$ there exist $\tbinom{M}{K}$
possible subsets of cardinality $K$. As already mentioned,  
 we shall not look for the sparsest representation 
 but make the search of the appropriate subspace 
 tractable by means of recursive greedy pursuit strategies, 
 which are only step wise optimal. Before discussing our approach  
 some considerations are in order. 
\section{Getting ready for the search} 
\label{sec4}
In this section we highlight some properties that will be of 
assistance in the next 
section, where we will present our strategy for the search 
of the sparse representation achieving the desired signal discrimination.
The goal is to avoid the computation of the measurement vectors 
in the whole subspace.
Instead, we strive to find the subspace $\SV_K \subset \SV$, where the
signal component one wants to extract from a signal $| f \ra$
is assumed to lie.  
We work under the hypothesis that the subspace $\SWC$ is given and
fixed. Furthermore, $\SV \cap \SWC = \emptyy$,
which implies that there exists a unique solution for the signal splitting.
The problem we need to address arises from the fact that, if 
the subspaces $\SV$ and $\SWC$ are not well separated,  
the numerical calculation of the measurement vectors 
is not accurate (due to the numerical operations being carried 
out in finite precision arithmetic). As a consequence, the  
representation of the corresponding projector fails to produce the 
correct signals separation. This effect is very much magnified 
if the data are affected by errors no matter how insignificant 
those errors are.

Assuming that we are able to accurately compute in finite 
precision arithmetic $r$ measurement vectors, we could attempt 
to single out a signal belonging to a subspace spanned by at most $r$ 
vectors (i.e. we could attempt to separate from $| f \ra$ 
 a signal expressible as in (\ref{mod2}) but at most with 
$r$ nonzero coefficients). However, as discussed above, 
even possessing this knowledge about 
the sought signal the problem of finding the right subspace by 
exhaustive search is not 
affordable. Hence, an adaptive greedy strategy for the
subspace selection, given a signal, was advanced in \cite{Reb07b}. 
Before revising and extending that strategy we need to recall two 
relevant properties of oblique projectors.
\begin{property}
\label{pro1}
The oblique projector $\EVW$ satisfies
$\op_{\SW} \EVW= \op_{\SW}$.
\end{property}
\begin{proof} 
It readily follows by applying $\op_{\SW}$ on both sides of 
\eqref{yoni} or \eqref{obli2}. 
Since $ \op_{\SW} | v_i \ra = \kui$ and 
$\la u_i,v_j\ra = \la u_i, u_j\ra$, one has
\be
\label{pw}
\op_{\SW} \EVW= \sum_{i=1}^M  | u_i \ra \la w_i |= 
\hat{U} \hat{G}^\dagger \hat{U}^\ast=\op_{\SW}.
\ee
%Moreover, since $\op_{\SW} \eta_i= \xi_i$, considering \eqref{evw2} 
%we have
%\be
%\label{pw2}
%\op_{\SW} \EVW = \sum_{i=1}^N \xi_i \la \xi_i , \cdot \ra= \op_{\SW}.
%\ee
\end{proof}
\begin{property} 
\label{orcom}
Given a signal $ | f \ra$ in 
$\SV + \SWC= \SW \oplus\SWC$, the only vector $ |g \ra\in \SV$
satisfying 
\be
\label{co2}
\op_{\SW} | f \ra = \op_{\SW} |g\ra
\ee
is $| g \ra= \EVW  | f \ra$.
\end{property}
\begin{proof}
If $|g\ra = \EVW  | f \ra $ \eqref{co2} trivially follows from Property 
\ref{pro1}.
Let us assume now that there exists $|g\ra \in \SV$ such that  
\eqref{co2} holds.
Then $\op_{\SW} (| f \ra -| g \ra)=0$, i.e.,
$( |f \ra- | g \ra )\in \SWC$. Hence $ \EVW( | f \ra- | g \ra)=0$ and,
since $|g  \ra \in \SV$,  this implies that $\EVW  | f \ra = | g \ra $.
\end{proof}
Let us suppose that $\SV_k= \Spann\{|v_i\ra\}_{i=1}^k$ is given and 
the spanning set is linearly independent. Assuming that
$\SV_k \cap \SWC = \emptyy$ we guarantee that the set of 
vectors $\{|u_i\ra \}_{i=1}^k$, with $|u_i\ra$ given in 
\eqref{ku}, is also linearly independent.  Consequently the dimension 
of $\SV_k$ is equal to the dimension of 
$\SW_k= \Spann \{|u_i\ra\}_{i=1}^k= \Spann{\{|\wtik \ra\}_{i=1}^k}$. 
We use now a superscript $k$ to indicate that the measurement 
vectors $\{|\wtik\ra\}_{i=1}^k$ span $\SW_k$. 
Hence these vectors give rise to  
the oblique projection of a signal $ | f \ra$, onto $\SV_k$ and
along $\SWC$, as given by:
\be
\label{ato}
\EVKW  | f \ra =\sum_{i=1}^k |v_i \ra \la  \wtik | f \ra = 
\sum_{i=1}^k c_i^k \kvi.
\ee
It is clear from (\ref{ato}) that if the atoms in the  atomic 
decomposition were to be changed (or some  
atoms were added to or deleted from the decomposition) 
the measurement vectors $ \{|\wtik\ra\}_{i=1}^k$, and consequently the coefficients $\{\cik \}_{i=1}^k$
in (\ref{ato}), would need to be modified. The recursive equations 
below provide an effective way of implementing the task.\\ 

{\em{Forward/backward adapting of measurement vectors}}\\

Starting with $|\wt_1^1\ra=\frac{|u_1\ra}{|||u_1\ra||^2}$, 
and $|u_1\ra$ as in \eqref{ku},
the measurement vectors $ \{|\wt_{i}^{k+1}\ra\}_{i=1}^{k+1}$ can be recursively 
constructed from $\{|\wtik\ra \}_{i=1}^{k}$ as follows \cite{Reb07a}:
\ben
\label{eq}
|\wtikp \ra &=& | \wtik \ra  - |\wt_{k+1}^{k+1}\ra \la u_{k+1} | \wtik \ra,
\quad i=1,\ldots,k \label{wik}\\
|\wt_{k+1}^{k+1}\ra&=& \frac{|\gamma_{k+1}\ra}{|||\gamma_{k+1}\ra||^2},\quad
|\gamma_{k+1}\ra= |u_{k+1}\ra-\op_{\SWK}|u_{k+1}\ra, \label{wkkp}
\een
where $\op_{\SWK}$ is the orthogonal projector onto $\SWK=
\Spann\{| u_i \ra \}_{i=1}^k$. 
We note that, since $|u_{k+1}\ra = \op_{\SW} |v_{k+1} \ra$  and 
$\op_{\SW} |\wtik \ra=  | \wtik \ra $, 
\eqref{wik} can also be written as
\ben
| \wtikp \ra &=& |\wtik \ra  -  |\wt_{k+1}^{k+1} \ra  \la v_{k+1}| \wtik \ra,
\quad i=1,\ldots,k. \label{wik2}
\een
It follows from the above equations that when incorporating a linearly
independent atom 
$|v_{k+1}\ra$ in the atomic decomposition \eqref{ato}, the coefficients 
can be conveniently modified according to the recursive equations
\ben
\label{recf}
c_{k+1}^{k+1}&=& \la  \wt_{k+1}^{k+1} | f \ra, \\
c_i^{k+1}&= &\la \wtikp | f \ra= c_i^{k} - c_{k+1}^{k+1} \la \wtik| v_{k+1} \ra,
\quad i=1,\ldots,k. \label{wik22}
\een
Conversely, considering that the atom, $|v_j\ra$ say, is to be removed 
from the atomic decomposition \eqref{ato}, and 
denoting the corresponding subspaces 
$\SVKJ$ and $\SWKJ$, in  
order to span $\SWKJ$
the measurement vectors $ \{|\ut_i^{\kj}\ra \}_{i=1, i\ne j}^k$ are modified 
according to the equation \cite{Reb07a}
\be
\label{duba}
|\ut_i^{\kj}\ra=|\utik\ra-\frac{|\ut_j^k\ra \la \ut_j^k |\utik\ra}{|| |\ut_j^k\ra ||^2},\quad
i=1,\ldots, j-1, j+1, \ldots, k.
\ee
Consequently, the coefficients in \eqref{ato} should be changed to
\be
\label{coba}
c_i^{\kj}=\cik-\frac{\la \utik | \ut_j^k \ra c_j^k}{|||\ut_j^k\ra||^2 },\quad
i=1,\ldots, j-1, j+1, \ldots,k.
\ee
%Each time a new vector
%$v_{\ll_{k+1}}$ is added to the set ${\SV}_k$ to form the set ${\SV}_{k+1}$,
%we update the biorthogonal set
%$\{w^{k+1}_i,\,i=1\cdots,k+1\}$ by using the equations, see \cite{Reb06}:
%\begin{eqnarray}\label{dualupdate}
%w_{k+1}^{k+1}&=&\frac{\gamma_{\ll_{k+1}}}{\|\gamma_{\ll_{k+1}}\|^2};\;
%\gamma_{\ll_{k+1}}=u_{\ll_{k+1}}-P_{\SWK}u_{\ll_{k+1}},\\
%w_i^{k+1}&=&w_i^{k}-w_{k+1}^{k+1}\la u_{\ll_{k+1}},w_i^{k}\ra,\,i=1,\cdots,k,
%\end{eqnarray}
%where $P_{\SWK}$ is the orthogonal projection onto $\SWK$.
\section{Adaptive pursuit strategy for subspace selection} 
\label{sec5}
Given a signal $| f \ra$, we aim at finding the
subspace $\SV_K \subset \SV$ 
where one of the signal components lies. Let us stress once again that
the problem arises from the impossibility of correctly computing the 
measurement vectors spanning the whole subspace $\SW$.  
Moreover, we have to face the fact that the corresponding 
signal component we want to represent {\em is not 
available}. What we know is that the available  signal, 
$|f \ra$, is expressible as the sum of two components 
$|f \ra = \kfV + \kfWC$ and that there exists an unknown 
subspace ${\SV_K} = \Spann \{ {|v_{\ell_{i}}\ra}\}_{i=1}^K 
\subset \SV$ 
where  $\{{\ell_{i}}\} _{i=1}^K$ is a set of $K$ unknown 
indexes such that
$ | f \ra= \kfV + \kfWC= |f_{\SVKF}\ra + \kfWC,$
with
\be
\label{fsu}
|f_{\SV}\ra= |f_ {\SVKF}\ra =\sum_{i=1}^K |v_{\ell_i}\ra \la w_i^K | f \ra.
\ee
Hence, if the set of indexes  $\{{\ell_{i}}\} _{i=1}^K$
were given, 
one could construct 
the measurement vectors $|w_i^K\ra$ in 
${\SW_K}= \Spann \{\op_{\SW} {|v_{\ell_i}}\}_{i=1}^K$ and 
obtain the component $\kfV= |f_ {\SVKF}\ra$ from \eqref{fsu}. 
Unfortunately, in the problem we are addressing  
neither the set of indexes  $\{{\ell_{i}}\} _{i=1}^K$ 
nor the component $\kfV$ are given. Nevertheless, by applying 
 $\op_{\SW}$ to  both sides of \eqref{fsu} we obtain
\be
\label{fsu2}
  |f_\SW\ra = |f_{\SW_K}\ra= \sum_{i=1}^K |u_{\ell_i}\ra \la w_i^K | f \ra.
\ee
Denoting $\hat{I}_{\SS}$ to the identity operator in $\SS$ we have
$\op_{\SW} = \hat{I}_{\SS}- \op_{\SWC}$. Thus, since the subspaces 
$\SS$ and $\SWC$ are known, we do have access to the component 
$|f_{\SW}\ra$. We can then look for the set of indexes  
$\{{\ell_{i}}\}_{i=1}^K$ 
to approximate this component as in (\ref{fsu2}). 
\begin{remark}
Notice that the measures $\la w_i^K | f \ra$
involved in \eqref{fsu} and \eqref{fsu2} are the same. Therefore,
by finding the representation of $|f_{{\SW}_K}\ra$ we have the information
which is needed to obtain $|f_{\SVKF}\ra$  from \eqref{fsu}. Let us 
stress that the need to deal with $|f_{{\SW}}\ra$ 
also introduces the bad conditioned nature of the problem we are considering. 
Certainly,  having access to the signal $|f_{\SV}\ra$ would imply that,
provided that the spanning set $\{|v_i\ra\}_{i=1}^M$ were well conditioned, 
one could find the sparse approximation  \eqref{fsu}  
without difficulty. However, even  when the conditioning of this spanning set 
is ideal (i.e. $\{|v_i\ra\}_{i=1}^M$ is an orthonormal basis for ${\SV}$) 
the fact that we  need to deal with the projection $|f_{\SW}\ra$, 
which is sparse in the set $\{|u_i\ra= \op_{\SW}|v_i\ra\}_{i=1}^M$,
introduces the difficulty we have to face. 
If the set $\{u_i\}_{i=1}^M$  were well 
conditioned, the robust signal splitting  could be obtained by a simple 
projection. The situation we are concerned with comprises the cases 
in which the whole set $\{u_i\}_{i=1}^M$ is very bad conditioned but 
the projection  \eqref{fsu} has a well conditioned representation 
in the subset $\{u_{\ell_i}\}_{i=1}^K \subset \{u_i\}_{i=1}^M$  
we aim to find.
\end{remark}
The proposed strategy for selecting the subset of atoms 
$\{|u_{\ell_i}\ra\}_{i=1}^K$ evolves by 
stepwise selection and is in 
line with the strategies in 
\cite{RP02a, RP02b, RP04}.
By fixing $\op_{\SWK}$, at iteration $k+1$ we select the index 
$\ell_{k+1}$ such that 
$||\op_{\SW} |f \ra - \op_{\SWKP}| f \ra ||^2$ is minimized. 
\begin{proposition} 
Let us denote by $J$ the set of indices $\{\ell_1,\ldots,\ell_k\}.$
Given $\SWK= \Spann\{|u_{\ell_{i}}\ra\}_{i=1}^k$, 
the index $\ell_{k+1}$ corresponding to the atom 
$|u_{\ell_{k+1}}\ra$ 
for which $||\op_{\SW}| f \ra  - \op_{\SWKP}|f\ra ||^2$ is minimal
is to be determined as
\be
\label{oomp}
 \ell_{k+1}= \arg\max\limits_{n \in J \setminus J_{k}}
\frac{|\la\gamma_n|f\ra|}{\||\gamma_n\ra\|},\, 
{\||\gamma_n\ra\|} \neq 0,
\ee
with $|\gamma_n\ra$ given in \eqref{wkkp}, and $J_{k}$ the set 
of indices that have been previously chosen to determine $\SWK$.
\end{proposition}
\begin{proof}
It readily follows  since $\op_{\SWKP}  |f \ra= \op_{\SWK}  |f \ra+ 
\frac{| \gamma_n \ra \la\gamma_n |f\ra}{\|\gamma_n\|^2}$ and hence
$||\op_{\SW}  | f \ra - \op_{\SWKP}  | f \ra  ||^2= 
||\op_{\SW}  | f \ra ||^2 - || \op_{\SWK}  | f  \ra||^2 -
\frac{|\la \gamma_n  | f  \ra |^2}{\|\gamma_n\|^2}   $.  Because  
$\op_{\SW}  |f \ra $ and $\op_{\SWK} | f \ra $ are
fixed, $||\op_{\SW} | f \ra -\op_{\SWKP} | f \ra ||^2$ 
is minimized if 
$\frac{|\la\gamma_n|f\ra|}{\||\gamma_n\ra\|},\, \| \gamma_n \| 
\neq 0$ is maximal over all $n\in J\setminus J_{k}$.
\end{proof}
The OBMP selection criterion  
\cite{Reb07b}  
  selects the index 
$\ell_{k+1}$ as the maximizer over $n \in J\setminus J_{k}$ of
$$\frac{|\la\gamma_{n} |f\ra|}{\||\gamma_{n}\ra\|^2},\quad ||\gamma_n||\neq 0.$$
This condition was  proposed  in \cite{Reb07b} based on 
the {\em consistency principle} \cite{UA94,Eld03}.
 Such a  principle, introduced in  \cite{UA94} and extended 
in \cite{Eld03},
states that the reconstruction of a signal should be 
self-consistent in the sense that, if the approximation is
measured with the same vectors, the same measures should be obtained.
Accordingly, the above OBMP criterion was derived in \cite{Reb07b} 
in order  to select the measurement 
vector  $|\wtkp_{k+1} \ra$
producing the maximum {\em {consistency error}} 
$\Delta= |\la \wtkp_{k+1}|f - \EVKW f \ra|$, with regard to 
a new measurement $|\wtkp_{k+1}\ra$.  However, since the measurement 
vectors are not normalized to unity, it is sensible to consider 
the consistency error relative to the corresponding 
vector norm $|||\wtkp_{k+1}\ra||$, and select the index  
so as to maximize over $k+1 \in J \setminus J_{k}$ 
the {\em{relative consistency error}} 
\be
\label{reler}
\til{\Delta}= \frac{|\la \wtkp_{k+1}| f - \EVKW f \ra|}{|||\wtkp_{k+1}\ra||}, 
\quad |||\wtkp_{k+1}\ra|| \ne 0.
\ee
In order to cancel this error, the new approximation is constructed accounting 
for the concomitant measurement vector. 
\begin{property}
The index $\ell_{k+1}$ satisfying (\ref{oomp}) 
maximizes over $k+1 \in J\setminus J_{k}$ the 
relative consistency error \eqref{reler}
\end{property}
\begin{proof}
Since for all vector $|\wtkp_{k+1}\ra$ given in 
\eqref{wkkp} $\la \wtkp_{k+1} | \EVKW =0$ and 
$|||\wtkp_{k+1}\ra||= |||\gamma_{k+1}\ra||^{-1}$ 
we have
$$\til{\Delta}= \frac{|\la \wtkp_{k+1}| f \ra |}{||\wtkp_{k+1}||}= 
\frac{|\la \gamma_{k+1} | f \ra |}{||\gamma_{k+1}||}.$$
Hence, maximization of $\til{\Delta}$ over $k+1 \in J\setminus J_{k}$ 
is equivalent to (\ref{oomp}). 
\end{proof}
It is clear at this point that the forward selection of 
indices prescribed by proposition (\ref{oomp}) is equivalent to 
selecting the indices by applying  the subset selection criterion 
introduced in \cite{RP02a} (c.f. Theorem 1) on the projected signal 
$\op_{\SW}|f\ra$ using the dictionary $\{|u_i\ra\}_{i=1}^M$. 
In the context of subspace selection for signal representation  
we have termed such a criterion Optimized Orthogonal 
Matching Pursuit (OOMP)\cite{RL02}. 

The hypothesis that the computation of more than $r$ measurement 
vectors becomes an ill posed problem enforces the OOMP
selection of indices to stop if iteration $r$ is reached.
Nevertheless, the fact that 
the signal is assumed to be $K$-sparse, with $K \le r$, does not
imply that before (or at) iteration $r$ one will always 
find the correct subspace. 
The $r$-value just indicates that it is not possible to 
continue with the forward selection, because 
the computations would become inaccurate and unstable. 
Hence, if the right solution was not yet found, one needs 
to implement a strategy accounting for the fact that 
it is not feasible to compute more than $r$ measurement vectors.  An 
adequate procedure is achieved by means of the swapping-based 
refinement to the OOMP approach introduced in 
\cite{AR06}. As discussed below, it consists of interchanging 
already selected atoms with nonselected  ones.

Consider that at iteration $r$ the correct subspace has not 
appeared yet and the selected indices are labeled by the $r$
indices $\ell_1,\ldots,\ell_r$. 
In order to choose the index of the atom that 
minimizes the norm of the residual error as passing 
from approximation $\op_{\SWKr}|f\ra$ to approximation 
$\op_{\SWKrJ}|f\ra$ 
we should fix the index of the atom to be deleted, $\ell_j$ say,
as the one for which the quantity
\be
\label{swab}
\frac{|c_i^r|}{|||w_i^r\ra||},\,i=1,\ldots,r.
\ee
is minimized \cite{Reb04,RP04,ARS04,AR06}.

The process of eliminating one atom from the atomic decomposition 
\eqref{ato} is called {\em backward step} while the  process of 
adding one atom is called {\em forward step}. 
The forward selection criterion to choose the atom to replace the 
one eliminated in the previous step is accomplished by 
 finding the index $\ell_i, i=1,\ldots,r$ for which   
 the functional 
\be
\label{swaf}
e_n= \frac{|\la\nu_n | f \ra|}{|| | \nu_n \ra ||},\quad 
\text{with}\quad |\nu_n \ra = |u_n \ra - \op_{\SWKrJ} |u_n\ra
,\quad |||\nu_n\ra||\neq 0
\ee
is maximized. 
In our framework, using (\ref{duba}), the projector 
$\op_{\SWKrJ}$ is computed as 
$$\op_{\SWKrJ} = \op_{\SWKr} - 
\frac{\la \ut_i^r | \ut_j^r\ra \la \ut_j^r |}{|||\ut_j^r\ra||^2}.$$
Since $\op_{\SWKr}$ and $|\ut_j^r\ra$ are available, the 
computation of the sequence $|\nu_n\ra$ in \eqref{swaf}
is a simple operation. 

The swapping of pairs of atoms is 
repeated until the swapping operation, if carried out, would not 
decrease the approximation error. The implementation details 
for an effective realization of this process are given in \cite{AR06}, 
and MATLAB codes are available at \cite{Webpage}.  
Since there is no guarantee that at the end of the swapping 
of pairs of atoms the correct subspace has been found, the
process can continue by increasing the number of 
atoms the swapping involves. At the second stage, in line with 
\cite{AR06a},
 we propose the swapping to be realized by the combinations 
 of two backward steps followed by 
two forward steps, provided that the interchange of the two 
atoms improves the approximation error. If at the end of the
second stage the right subspace has not yet been found, the  number of 
atoms involved in the swapping is increased up to three 
and so on. Notice that 
if the number of atoms to be interchanged reaches the value $r$ the 
whole process would repeat identically. This is avoided by 
initiating the new circle with a different initial atom. 
%Although convergence cannot be guaranteed, the above specified  
The above specified hypothesis ensures that the algorithm 
will stop when the correct signal splitting has been found. 
At such a stage one has
$\op_{\SW} |f\ra = \op_{\SW_r} |f \ra$ with $\SW_r$ 
spanned by the selected atoms 
$\ell_1,\ldots,\ell_r$. If the order $K$ of sparseness 
of the signal is 
less than $r$ a number of $r-K$  coefficients in the atomic 
decomposition 
$$|f_{\SV_r}\ra =\sum_{i=1}^r |v_{\ell_i}\ra \la w_i^r|f \ra = 
\sum_{i=1}^r c_i^r |v_{\ell_i}\ra$$
will have zero value.

In the realistic case where the measurements are affected by errors, 
the proposed iterative process is to be stopped when the condition 
$$||\op_{\SW} | f \ra - \op_{\SW_r} | f\ra || \le \delta$$ is reached 
(where $\delta$ should be determined by taking into account the errors of 
the data).
\section{Numerical simulations}
\label{sec6}
\subsection{Impulsive noise filtering}
We extend here the example in \cite{Reb07a} concerning the 
elimination of impulsive noise from the
register of the motion of uncoupled damped harmonic oscillators. 

The $n$-th oscillator is characterized by a
frequency of $\frac{n}{2}$ Hertz  and  
its motion,  as a function of time,
is given by the equation  
\be
\label{xn}
\la t | x_n \ra= x_n(t)= e^{-t} \cos (\pi n t).
\ee
In \cite{Reb07a} the register of system's motion  was considered to 
be the signal
\be
\label{premo}
\la t | f_{{\SV}_{100}} \ra = \sum_{n=1}^{100}
\frac{e^{-t} \cos (\pi n t)} {(1+0.7(n-75)^2)},\quad t \in [0,1]
\ee
and the impulsive noise corrupting  this signal 
was assumed to belongs to the subspace
\be
\label{swc}
\SWC= \Spann \{e^{-100000(t-0.0025j)^2},\, j= 1,\ldots,400,\,
t \in [0,1]\}.
\ee
Due to the nature of the  distribution of frequencies in 
\eqref{premo}, the corresponding oblique projector 
for filtering the 
impulsive noise from the register was recursively constructed
by incrementing the frequency $n$ one by one, 
until the projection of the noisy signal 
onto the span of the set $\{x_n(t)\}_{n=1}^k$
remained unaltered by increasing the number $k$ of 
elements in the set up to some value.
As remarked in \cite{Reb07a},  
such a procedure is not always feasible.
Let us consider for instance  that the
100 oscillators have frequencies which are not restricted 
to the range $[1,100]$  but can be any integer number in the 
$[1,405]$ interval. In this case the recursive construction of the 
projector by incrementing the frequency one by one is in 
general not possible, because the corresponding numerical calculations
become ill posed before reaching the frequency $n=405$.
However, the oblique projector, along $\SWC$ given above, 
onto a subspace spanned by only 100 functions $x_n(t),\,t \in [0,1]$
 with $n\in [1,405]$, can be accurately calculated. Thus, by applying 
the techniques of the previous sections to find the right 
frequencies of the system, the cancellation of the 
impulsive noise in the register is rendered possible. This is 
illustrated by the numerical simulation described below. 
 
The impulsive noise is simulated by randomly taking 
200 pulses in \eqref{swc}. The system's motion is  simulated  as a linear 
combination of 100 functions \eqref{xn} the frequencies of 
which are taking, randomly, from the set $[1,405]$. The 
data are assumed to be known in single precision. 
The simulation was run 50 times and in all the cases the cancellation 
of the impulsive noise was successful. The left graph of Fig.~1 
plots one of the realizations of the experiment  (motion of the  system 
plus impulsive noise vs time). The graph on the right depicts the  
result obtained by applying the proposed technique 
to the signal on the left (it coincides with the 
line representing the true signal). On the contrary, although 
the spaces \eqref{swc} and $\Spann\{x_n(t),\, t\in [0,1] \}_{n=1}^{405}$  
are `theoretically' complementary, since the 
construction of the corresponding oblique projector is very ill posed, 
the projection fails to correctly separate the signals.

\begin{figure}[!ht]
\label{f2}
\begin{center}
\includegraphics[width=7.5cm]{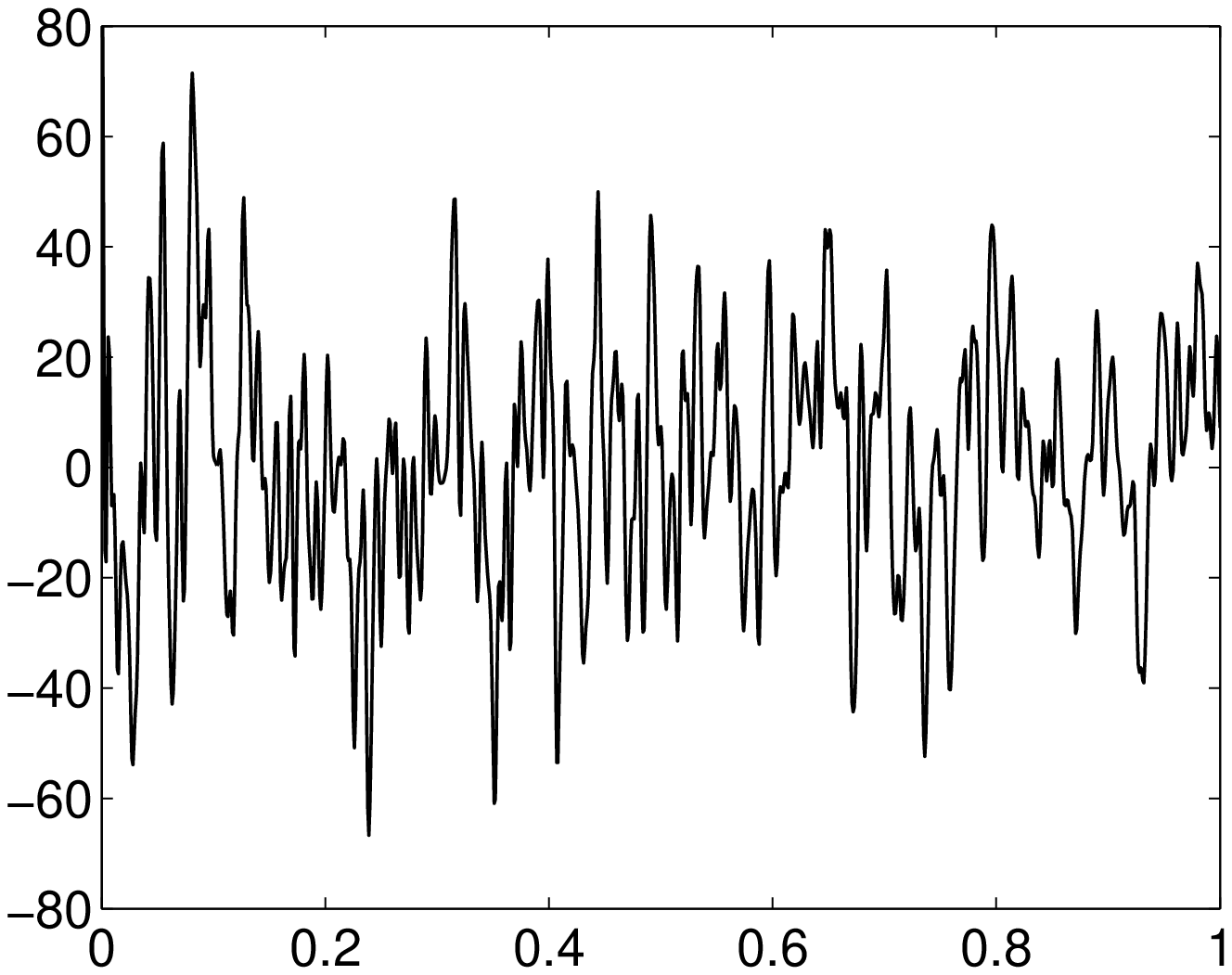}
\includegraphics[width=7.5cm]{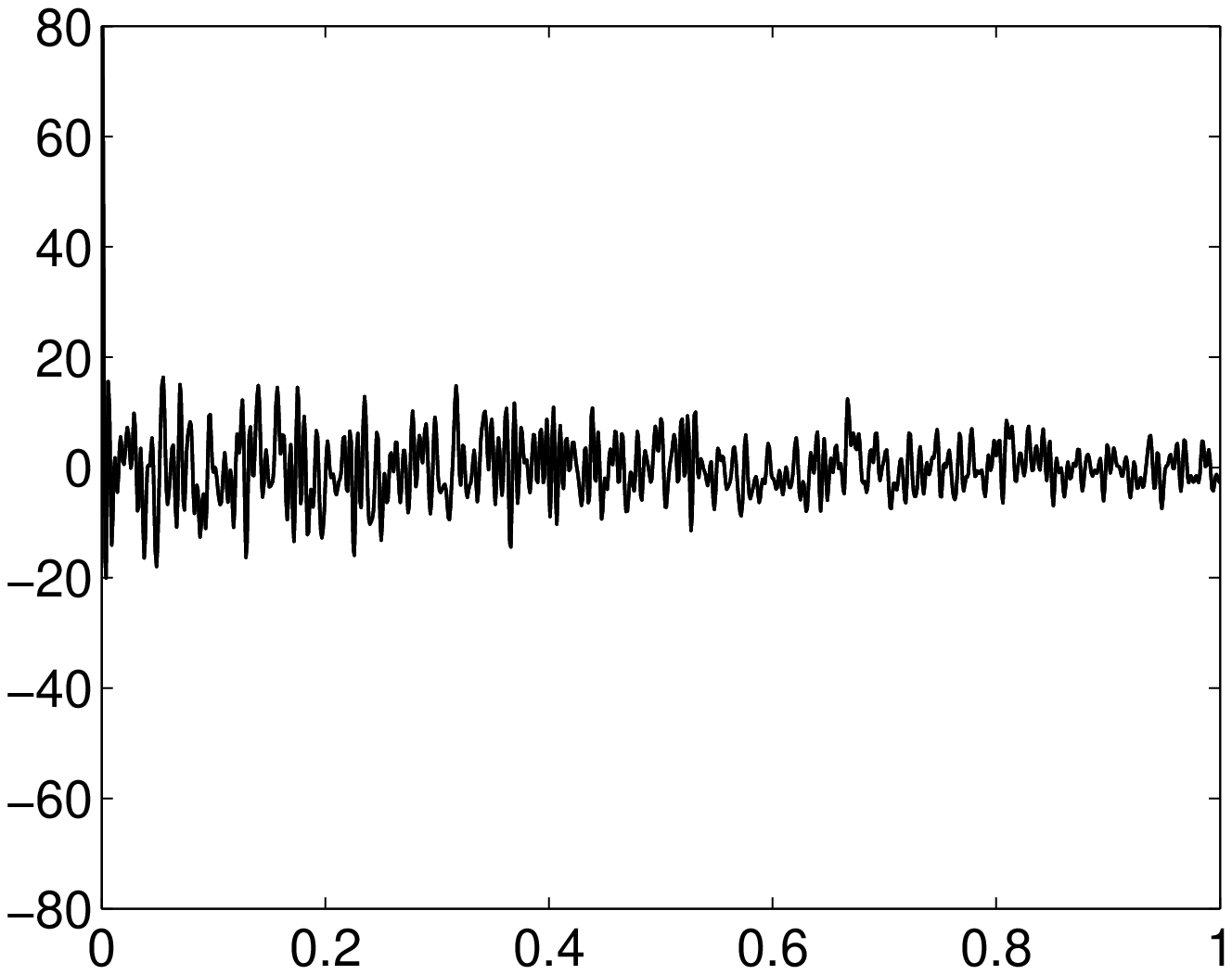}\\
\end{center}
\caption{The left graph depicts the register of
the motion of a system 
consisting of 100 damped harmonic oscillators, whose 
frequencies are integer numbers randomly taken from the interval
$[1,405]$, corrupted by 200 pulses randomly taken from 
the subspace $\SWC$ given in \eqref{swc}. 
The graph on the right depicts the result of 
filtering the signal of the previous graph
by the strategy described in Section 5. The result coincides with the 
true motion of the simulated system.} 
\end{figure}

\subsection{Application to the separation of a spectrum 
from blackbody radiation background}
We consider an hypothetical situation where the 
background is assumed  to be produced by a linear 
combination of up to five blackbodies (e.g. stars) at temperatures of 
$T_1=3000$K, $T_2=3500 $K , $T_3=4000 $K, $T_4=4500 $K , and $ T_5=5000 $K.
Hence, in this case the 
subspace $\SWC$ is defined as
$\SWC= \Spann\{y_i\}_{i=1}^5$, where $y_i$ are functions 
of the wavelength $\lambda$ as given by
$$y_i(\lambda)= \frac{C_1}{\lambda^5(e^{\frac{C_2}{\lambda T_i}}-1)}, \quad  
C_1=3.7419 \times 10^{-6} {\rm{erg\, cm^2\, s^{-1}}}, 
C_2=1.4288 {\rm{cm\, K}}.$$
Any linear  combination  of these functions is an acceptable 
background. In our numerical experiment the background, $y(\lambda)$, is generated as 
$y(\lambda)= \sum_{i=1}^5 y_i(\lambda)$.

We simulate a spectrum, on a region of $\lambda$  ranging from 
zero to $3 \mu \rm{m}$, by considering that it belongs to 
the cardinal cubic spline space on the interval  
$[0, 3\mu \rm{m}]$ with separation $b=2^{-4} \mu \rm{m}$ 
between consecutive knots.
Such a space can be spanned by a B-spline basis  
arising by translating a prototype B-spline \cite{Sch81,AR05}.
Different spectra  are simulated by randomly drawing $K=70$ functions 
$\{v_{\ell_i}\}_{i=1}^{70}$, from the basis consisting of $483$
functions, to 
generate the decomposition $\sum_{i=1}^{70} c_i v_{\ell_i}$, with 
random coefficients $c_i \in [0,1],\,i=1\ldots,70$. A particular 
realization is depicted in the bottom graph of Fig.~2. The 
simulated available signal is obtained by adding the background $y(\lambda)$ 
given above to the spectrum, and perturbing each data point with a normal 
distributed error of variance corresponding to a percentage of each data value. 
Let us remark that, although the background is not 
`exactly' in the cardinal spline space of the spectrum, it 
has a very good representation in such a space.  Hence,  
 the calculation of an oblique projector onto the 
given spline space and along the space of the background is
 expected to be very badly conditioned. Indeed, 
for very small errors (variance of $10^{-6}\%$ of each data value) 
the separation of the spectrum from the background is not 
possible by an oblique projection onto the whole spline space. 
However, in a simulation of 100 different spectra, each of which consisting 
of 70 randomly taken B-spline functions,
the separation was successful in all the 
cases by applying the proposed approach. Thus, a more realistic situation 
was simulated by increasing the variance of the error up to $1\%$ of each data value. Also in 
this case the spectrum recovery was a complete success. 
In order to make evident 
the errors' effect, the variance was increased up to $5\%$ of each data value.
One of the simulations is plotted in the top graph 
of Fig.~3. The middle graph of the same figure 
shows both, the theoretical 
`hidden' spectrum and the one recovered by the proposed approach. 
This graph is meant to illustrate  a `typical result', as in a run of the
100 simulations  described above the reconstruction of the corresponding 
spectra was of similar quality.  
The visualization of the approximation quality is made clearer  in
the bottom graph of Fig.~3, 
where a portion of the previous graph
(corresponding to the interval $[0.5,1]$) is plotted. 
Here we can see that, as expected, the 
approximation (broken line) fails to reproduce the peaks of low intensity 
(the one at $0.65\mu \rm{m}$).
This is of course understandable by comparing the intensity of the 
spectrum with the intensity of the data on the $[0.5,1]$ interval. In this 
region, a variance of $5\%$ of intensity of each data point entails an 
uncertainty of more than one unit in the corresponding scale of 
intensity. Thus, one cannot  
expect to correctly spot peaks of intensity of the same order as the errors. 
On the other hand, 
some spurious peaks which are not in the true spectrum may also appear
(note the small peak of negative intensity). 
However, on the whole we can confidently assert that the recovery of the 
simulated spectra is satisfactory even for significant error level. 
It is pertinent to point out that if one wished to 
avoid negative intensities one could penalize the selection of 
measurements leading to such negative values. In our framework  
this is implementable in a straightforward manner. 
For instance, in the forward selection procedure,  
at iteration $k+1$ we select 
the index $\ell_{k+1}$ satisfying \eqref{oomp} and 
it follows from \eqref{ato} and \eqref{recf}  that 
the spectral intensity vector at this step is given as
\be
|f_{\SV_{k+1}}\ra= |f_{\SV_{k}}\ra - \hat{E}_{\SV_{k}} |v_{\ell_{k+1}}\ra
\frac{\la \gamma_{k+1}| f \ra}{||  |\gamma_{k+1}\ra ||^2} + |v_{\ell_{k+1}}\ra 
\frac{\la \gamma_{k+1}| f \ra}{||  |\gamma_{k+1}\ra ||^2},
\ee
where $|f_{\SV_{k}}\ra$ is the spectral intensity obtained in the 
previous iteration,  $\hat{E}_{\SV_{k}}$ the oblique 
projector onto the previously selected subspace, $|f\ra$ the 
observed data vector and $|\gamma_{k+1}\ra$ (constructed as in  \eqref{wkkp}) 
is to be  determined in the selection process of the index $\ell_{k+1}$ according to 
the prescription of Proposition 1. Thus, restrictions on $|f_{\SV_{k+1}}\ra$ 
can be incorporated by disregarding the selected indices 
yielding unacceptable values of $|f_{\SV_{k+1}}\ra$. 
In the example we are  discussing here, 
one can avoid negative values of intensity  by 
disregarding those indices which satisfying \eqref{oomp} do not 
fulfill the condition $\la \lambda |f_{\SV_{k+1}}\ra = f_{\SV_{k+1}}(\lambda)
\ge  0$ for the values of $\lambda$ being considered.  
With the incorporation of this constrain the small negative pick 
in the middle graph of Fig.~3 disappears and 
the whole approximation improves. However, in some other realizations 
of the experiment when introducing the possibility 
constrains some small spurious picks of positive intensity 
appear, yielding on the whole an approximation of quality
comparable with the unconstrained one.
The spurious small peaks (or the absence or small peaks present in the 
true spectrum) are consequence of the considerable uncertainty in the 
data. The fidelity of the approximation with the true 
spectrum (in all the realizations of the experiment)
is improved by reducing the error of the data. 
\section{Conclusion} \label{sec7}
The construction of measurement vectors specially  
designed for separating signal components produced by phenomena of 
different nature was discussed.  Assuming that 
the subspaces hosting the signal components are
given, the required 
measurement vectors  should 
yield an oblique projection along one of the  
subspaces and onto the other. 
Considerations were restricted to 
those cases for which such subspaces 
are theoretically complementary, yet very close to each other, 
so that the construction of the measurement 
vectors for the whole space renders  an ill posed problem.
A recursive  strategy for finding  the right subspace 
to achieve the desired signal separation was then
discussed. 
By recourse to  numerical simulations  it was illustrated 
that, provided that the signal is sparse
in a spanning set of the signal subspace, 
the required signal splitting 
may be achieved by means of adaptive greedy techniques capable of 
searching for the required subspace while maintaining stability 
in the calculations.  When tested in situations 
involving significant level of errors  the proposed technique produced
satisfactory results. Therefore,  
we are led to conclude that the framework
for measurement design advanced here
should be of assistance to a variety of applications 
where the discrimination of phenomena of different nature is required. 
\subsection*{Acknowledgements}
Support from EPSRC (EP$/$D062632$/$1) is acknowledged.

\newpage
\begin{center}
{\bf{Figures}}
\end{center}
\begin{figure*}[ht!]
\begin{center}
\includegraphics[width=10cm]{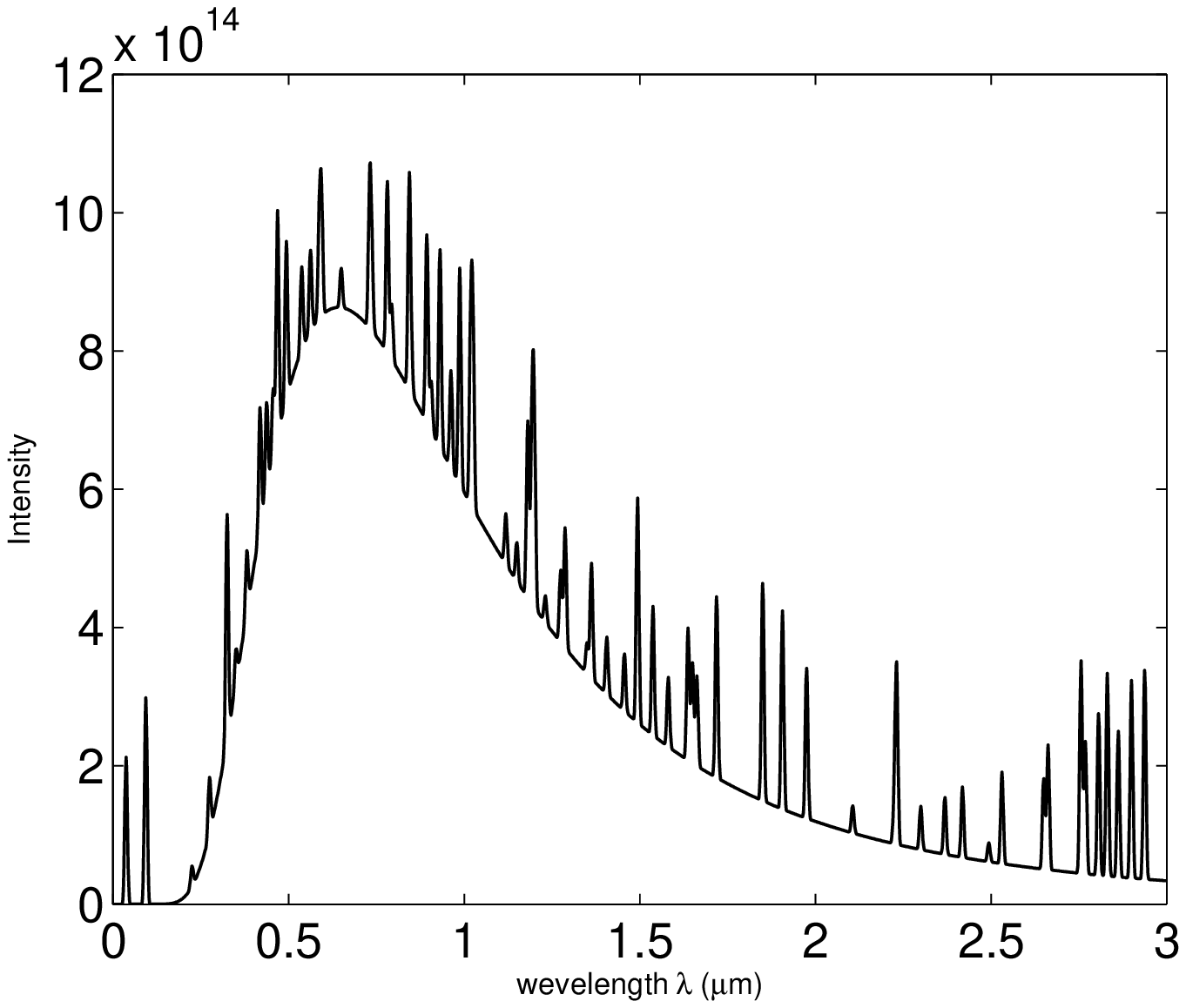}\\
\includegraphics[width=10cm]{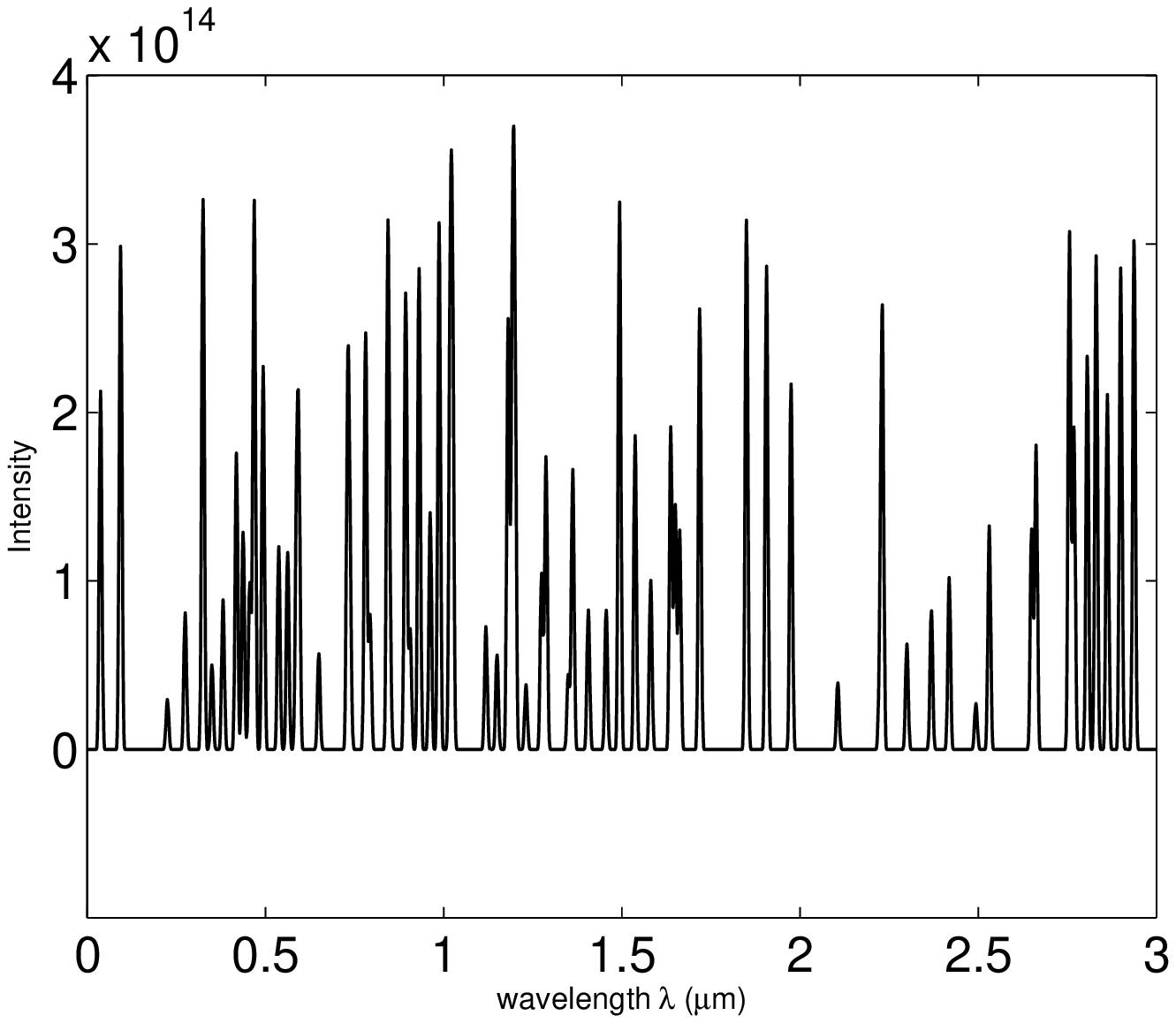}
\end{center}
\caption{The top graph depicts theoretical data produced
by the superposition of
the spectrum shown in the bottom graph and black body
radiation background.}
\end{figure*}
\newpage
\vspace{-1cm}
\begin{figure*}[ht!]
\begin{center}
\includegraphics[width=10cm]{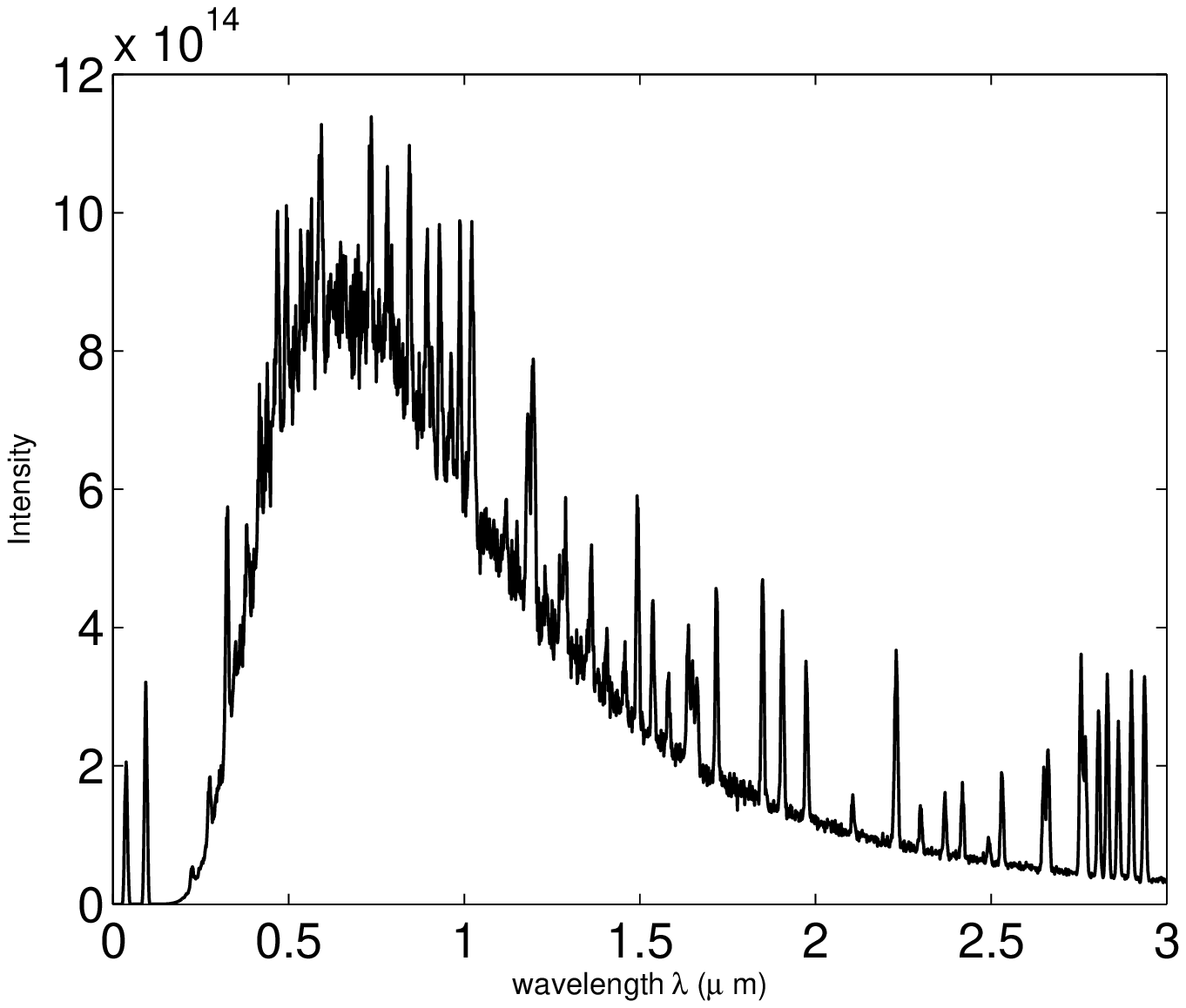}\\
\includegraphics[width=10cm]{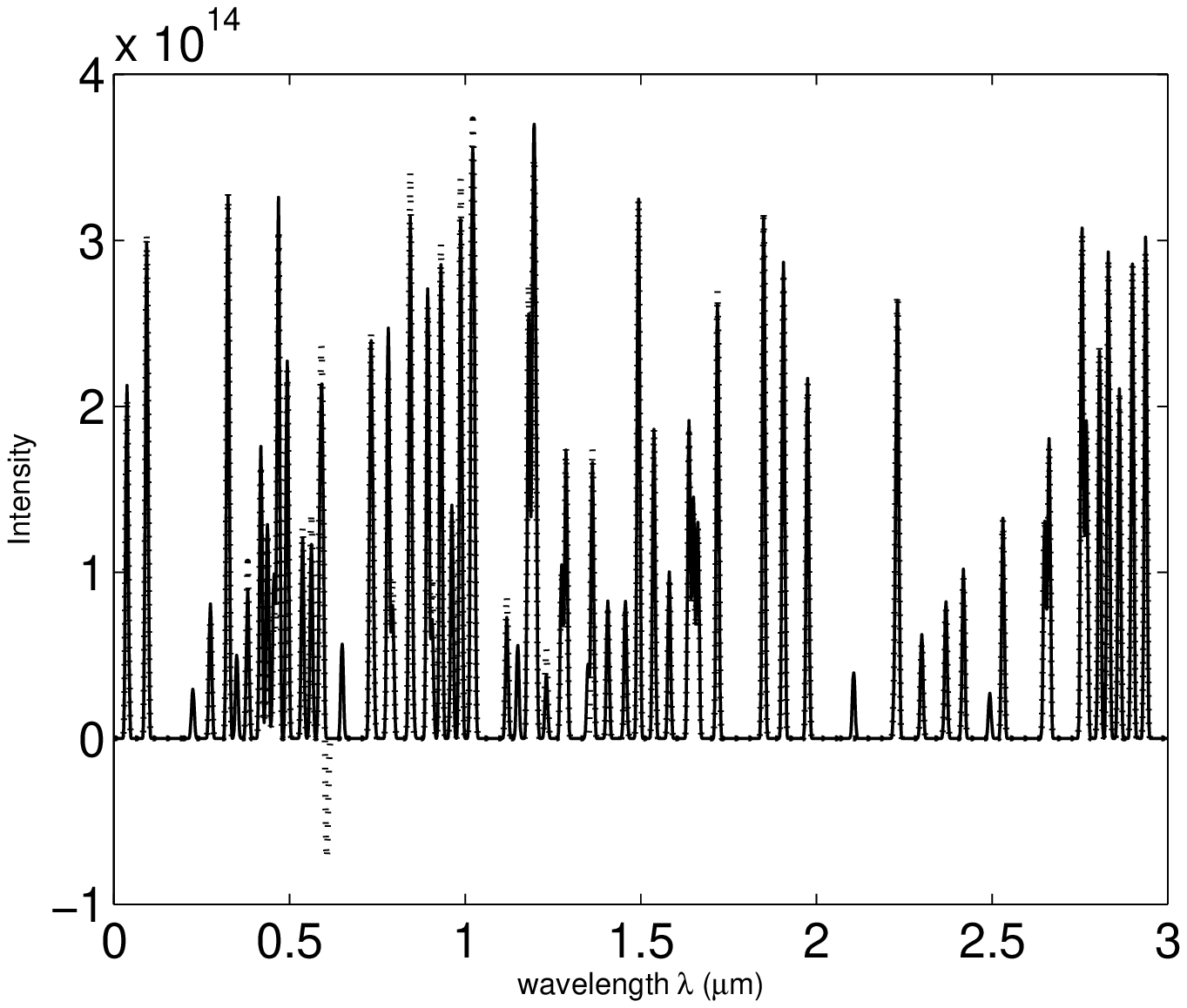}\\
\includegraphics[width=10cm]{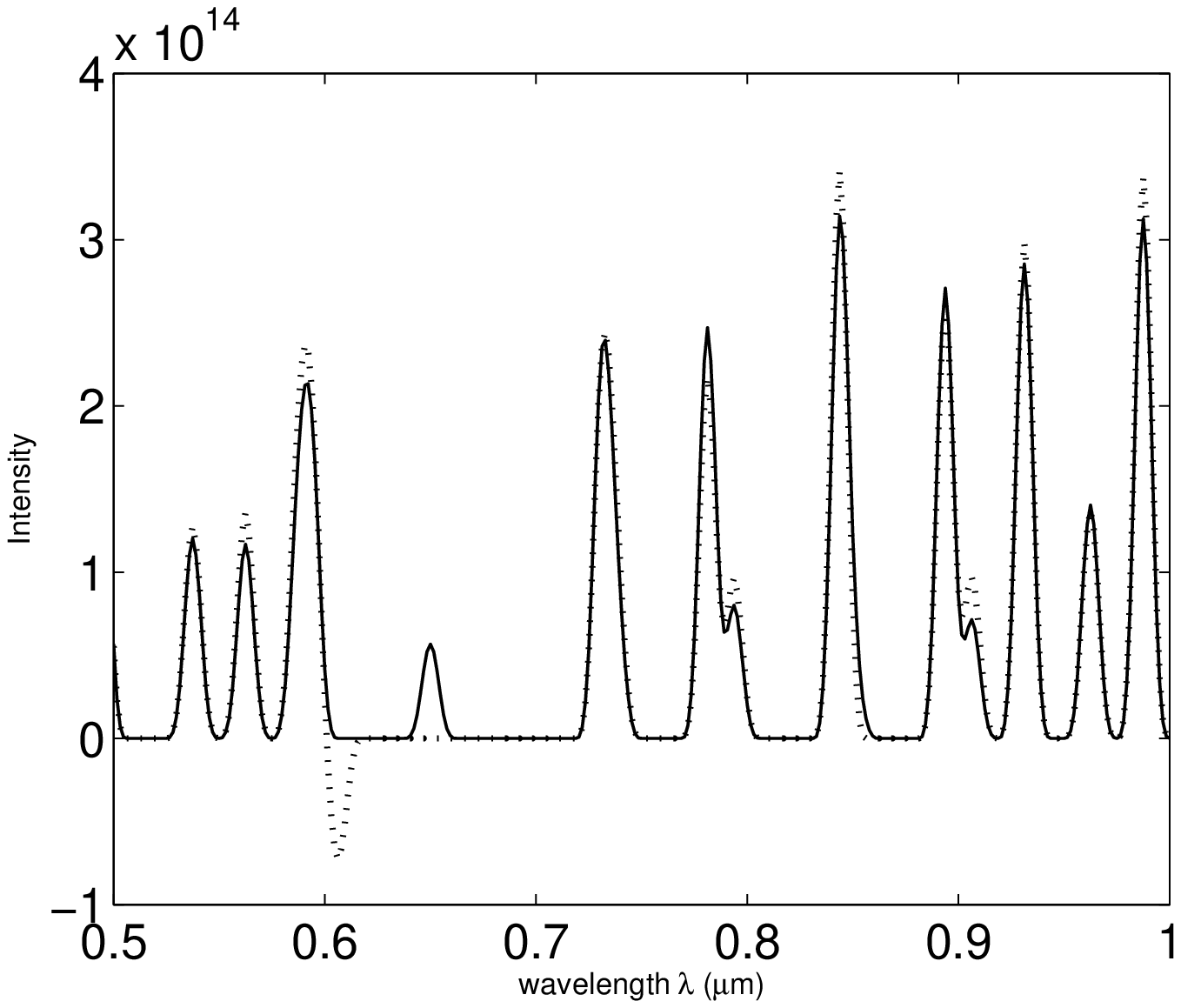}
\end{center}
%\caption{The top graph has the same description as the top graph of
%Fig.~1, but the data are perturbed with zero mean normal
%distributed errors of variance corresponding to $5\%$ of
%each data value. The
%continuous line in the middle graph represents the theoretical spectrum and
%the dotted line the spectrum obtained by the proposed approach from the
%data of the top graph. The bottom
%graph magnifies the region [0.5 1] in the previous graph.}
\end{figure*}
\newpage
\begin{center}
{\bf{Figure Caption}}
\end{center}
Figure 3. The top graph has the same description as the top graph of
Fig.~2, but the data are perturbed with zero mean normal
distributed errors of variance corresponding to $5\%$ of
each data value. The
continuous line in the middle graph represents the theoretical spectrum and
the dotted line the spectrum obtained by the proposed approach from the
data of the top graph. The bottom
graph magnifies the region [0.5 1] in the previous graph.\\

\newpage
\bibliographystyle{elsart-num}
\bibliography{revbib}
\end{document}